\documentclass[11pt,a4paper]{article}

\usepackage[margin=1.1in]{geometry} 
\usepackage[T1]{fontenc}
\usepackage[utf8]{inputenc}
\usepackage[UKenglish]{babel}
\usepackage{charter} 
\usepackage{microtype}
\usepackage{amsmath, amsthm, amssymb, amsfonts, mathtools}
\usepackage[style=numeric, backend=bibtex, sorting=none]{biblatex}
\usepackage{hyperref}
\usepackage[capitalize,noabbrev]{cleveref}
\usepackage{authblk} 
\usepackage{float}
\usepackage{array}
\usepackage{caption}
\usepackage[disable,colorinlistoftodos,textsize=footnotesize]{todonotes}

\newtheorem{theorem}{Theorem}[section]
\newtheorem{lemma}[theorem]{Lemma}
\newtheorem{proposition}[theorem]{Proposition}
\newtheorem{corollary}[theorem]{Corollary}
\theoremstyle{definition}
\newtheorem{definition}[theorem]{Definition}

\newtheorem{remark}[theorem]{Remark}
\newtheorem{claim}[theorem]{Claim}
\newenvironment{claimproof}[1][Proof of Claim]{%
  \begin{proof}[#1]%
}{%
  \end{proof}%
}

\newcommand{\domain}{\mathcal{U}}
\newcommand{\bintree}{\mathcal{T}}
\newcommand{\subtree}[1]{\bintree_{#1}}
\newcommand{\rootnode}{\text{root}}
\newcommand{\arr}{\text{Ar}}
\newcommand{\subseterr}{\subtree{err}}
\newcommand{\subseterrprime}{\widehat{\subseterr}}
\newcommand{\ok}{\operatorname{\mathsf{ok}}}
\newcommand{\error}{\operatorname{\mathsf{error}}}

\newcommand{\bisim}{\mathrel{\ooalign{$\leftrightarrow$\cr\hidewidth\raisebox{-.4ex}{\rule{0.8em}{.4pt}}\hidewidth}}}
\newcommand{\bd}{\mathsf{bd}}

\newcommand{\kvwin}{\mathsf{AWin}_\exists} 
\newcommand{\ekvwin}{\mathsf{AWin}_{\exists}^{\epsilon}} 
 
\newcommand{\bdwin}{\mathsf{Win}_\exists^ \mathsf{bd}} 

\newcommand{\power}{\mathcal P} 
\newcommand{\cat}[1]{\mathbf{#1}}

\newcommand{\set}{\cat {Set}} 

\newcommand{\model}{\mathcal{M}}
\newcommand{\logicalprop}{\varphi}
\newcommand{\forces}{\vDash}
\newcommand{\fun}{\mathsf{f}}
\newcommand{\V}{\mathcal V}
\newcommand{\evfun}{\mathsf{ev}_{\fun}}
\newcommand{\proj}[1]{\mathrm{proj}_{#1}}

\title{\textbf{Tree Automata Acceptance up to Measurable Defect}}

\author[1]{Anita Moyasari}
\author[1]{Harsh Beohar}
\author[1]{Charles Grellois}
\author[2]{Clemens Kupke}

\affil[1]{School of Computer Science, University of Sheffield, United Kingdom}
\affil[2]{University of Strathclyde, Glasgow, United Kingdom}

\date{} 

\begin{document}

\maketitle

\begin{abstract}
Automata acceptance can, in several situations of interest, be captured game-theoretically via \emph{acceptance games}.
The existence of a winning strategy for Verifier then captures the existence of a winning run-tree of a given automaton over a model.
However, such acceptance is rigid, in that it does not allow a measurable defect budget, which can be a challenge in software verification.
In this paper, we draw inspiration from how bisimulation distance can be defined as an extension of bisimilarity to define $\epsilon$-acceptance games. 
Our main theorem shows that a tree $\bintree$ is $\varepsilon$-accepted iff there is a tree $\bintree'$ that is accepted in the traditional (rigid) sense and the bisimulation distance of $\bintree'$ and $\bintree$ is at most $\varepsilon$.
Our work also suggests a strong connection with measure theory, of which we give a preliminary exploration via appropriate examples.
Our framework is defined over binary trees with leaves and infinite branches, and strictly contains the case in which binary nodes are seen as probabilistic choice and the defect measures the probability of the set of rejected branches.

\vspace{1.5em}
\noindent \textbf{2012 ACM Subject Classification:} Theory of computation $\to$ Automata extensions; Theory of computation $\to$ Verification by model checking

\end{abstract}

\section{Introduction}

Model checking is a successful approach to verification, in which a program's behaviour is modelled by a mathematical object $\model$ and the specification of the program is expressed as a logical property $\logicalprop$. In many situations of interest (such as modal $\mu$-calculus verification over trees, see for instance~\cite{JaninWlukiewicz:Automata}), the model-checking problem $\model \forces \logicalprop$ can be equivalently formulated by a two-player game played over the states of an automaton $\mathcal A_\logicalprop$ and the model $\model$, where Verifier is trying to assert that the property $\logicalprop$ is valid over $\model$, while Falsifier is dually trying to establish that $\logicalprop$ is invalid over $\model$.

Such games traditionally model a Boolean, \emph{rigid} notion of acceptance. In this paper, we build on an insight from process algebra \cite{Milner1989}, in which the notion of \emph{bisimulation} \cite{Park:Bisimulation1981} is considered a key equivalence on process terms.
It is also Boolean in nature since two pointed labelled transition systems are either bisimilar, or they are not; but this notion has been lifted to the notion of a \emph{bisimulation distance} which is well studied in quantitative settings like probabilistic generalisations of labelled transition systems \cite{10.1007/3-540-48320-9_19,10.1007/3-540-48224-5_35}.

In this paper, we extend the traditional, rigid game-theoretic automata acceptance games to a quantitative framework in which a defect budget is allowed. This defect budget provides an upper bound over some measure of the obstruction to acceptance -- that we make precise later in the paper. This turns the game into a quantitative one, in which Verifier plays moves that estimate upper bounds of the defect that would be faced later in a play, subject to thresholds coming from a global constraint. To ease reading, we take advantage of the well-known left-child, right-sibling encoding (cf.\thinspace \cite{TATA:2008}) and develop our approach on binary trees with leaves and infinite branches. The approach relies on a specification of how the defect budget is split in between both successors of an inner node; this notably covers cases in which the defect computes the probabilistic measure of the set of accepting branches, but the approach is more general as it is not restricted to linear combinations of errors.

Interestingly enough, our approach builds much more than a mere quantitative extension of automata acceptance games via the introduction of defects. Our main theorem (cf.\thinspace \cref{thm:main}) proves, under the restriction that the automaton is deadlock free, that when the defect is split between the successors of an inner (thus binary) node in a way that, in the coalgebraic world, is called a \emph{distance lifting}, then:
\[
\text{acceptance of a tree } \mathcal{T} \text{ with defect at most } \varepsilon
\]
is equivalent to the
\[
\text{existence of } \mathcal{T}'\text{, accepted in the traditional sense and at distance at most } \varepsilon \text{ of } \mathcal{T}
\]
where this distance -- \emph{bisimulation distance} -- can itself be captured by game-theoretic means.

We also explore the natural connection with measure theory via targeted examples, leaving the systematic treatment of the connection between distance liftings for certain classes of distance functions and measure theory for future work.

Overall, our work introduces a quantitative approach to verification, in which defect is admissible within certain thresholds, with key notions from the coalgebraic community. We present our results in a way that will make them amenable to a future coalgebraic treatment, covering in that way more data structures than trees, yet in a similar manner.

\paragraph*{Related Works}

\textbf{Quantitative Languages.} In~\cite{chatterjeeDH10}, Chatterjee, Doyen and Henzinger introduce quantitative languages of words. Such languages replace the Boolean value of whether a word $w$ belongs to $L$ by a real number; weighted automata are then used, and their runs output quantitative information. By contrast, our work does not amend trees or automata: by allowing a defect budget in the game, we capture the measure of an obstruction to acceptance; our main theorem that this defect coincides with bisimulation distance.

\smallskip
\noindent
\textbf{Probabilistic Model Checking.} In probabilistic model-checking (see for instance~\cite{Vardi85,Baier2008}), a core question is to compute the probability with which a probabilistic model such as a Markov chain satisfies a given logical specification. Our approach enables the verification of Markovian processes -- unfolded to a tree form, and for which the distance lifting would be chosen so as to encode the probability distributions. But our approach can also accommodate measure defects that do not come from (the encoding of) probabilistic systems, and is actually structurally rooted in the measure of a  defect to traditional acceptance rather than in the extension of automata and models with stochastic information.

\smallskip
\noindent
\textbf{Behavioural Metrics.} Bisimilarity was extended to bisimulation distance in the context of Markov processes, see~\cite{10.1007/3-540-48320-9_19,desharnaisGJP04,10.1007/3-540-48224-5_35},
and more recently to coalgebras \cite{baldanBKK18}. In the theory of coalgebras, there are two prevalent ways to define distance liftings: one generalises the Kantorovich lifting of distributions known as the codensity lifting \cite{codensity} of an endofunctor; the other \cite{Bonchi_et_al:UptoFibration} generalises the Wasserstein lifting of distributions known as the coupling-based lifting of an endofunctor in \cite{humeau_et_al:LIPIcs.CSL.2025.29}. In general, the two approaches to define lifting are not equivalent; however, for polynomial endofunctors (like the functor $F=\Sigma_0 \uplus \Sigma_2 \times \_ \times\_$ whose coalgebras are labelled binary trees) over sets it is known to be the same \cite{humeau_et_al:LIPIcs.CSL.2025.29}.

Our paper is part of our wider program on $\epsilon$-acceptance for data structures presented as coalgebras, extending coalgebraic automata theory~\cite{kupke-venema:coalg-aut-theory} and generalising the associated game semantics to a more quantitative setting. In the current paper we limit ourselves to binary trees and do not prove our main theorem (\cref{thm:main}) in its full generality (i.e. at the level of coalgebras), which is left for future.

\paragraph*{Organization}
We begin in \Cref{sec:prelim} by introducing the necessary preliminaries regarding trees, tree automata, and acceptance games.
In \Cref{sec:quantitative}, we generalize these definitions to a quantitative framework and introduce the bisimulation distance game.
Within this setting, we extend the classical notion of an acceptance game to that of $\epsilon$-acceptance, incorporating within the traditional acceptance game information about a tolerable defect budget. \Cref{sec:main thm} then establishes the formal connection between the quantitative and Boolean versions of these games. Finally, in \Cref{sec:optimal epsilon}, we discuss two applications, specifically focusing on the relationship between these games and measures on binary trees.

\section{Preliminaries}
\label{sec:prelim}
    The objective of this section is to recall the notion of acceptance games in the context of labelled binary trees from \cite{kupke-venema:coalg-aut-theory}. As mentioned in the introduction, and for the sake of readability, we restrict our attention to binary trees with leaves and infinite branches, taking advantage of the left-child, right-sibling encoding (cf.\thinspace \cite{TATA:2008}). Note that all the results of this paper can however be restated for ranked trees over a finite signature containing labels of finite arities other than 0 or 2.

\subsection{Labelled Binary Trees}
\label{subsec:trees}

\begin{definition}
A \emph{(binary) tree} is a subset $\domain$ of $2^*$ that is
\begin{enumerate}
\item prefix-closed: $\forall w,w'\in 2^*\quad (w\in \domain \land w' \leq w) \implies w' \in \domain$, and
\item sibling-closed: $\forall w \in \domain\quad w.0 \in \domain \iff w.1 \in\domain$.
\end{enumerate}
A \emph{labelled tree}, denoted $\bintree = (\domain,\ell)$, consists of a tree $\domain$, a set of labels $\Sigma = \Sigma_0 \uplus \Sigma_2$, and a labelling function $\ell \colon \domain \to \Sigma$ satisfying $\forall w \in \domain\quad \ell(w) \in \Sigma_2 \iff w.\{0,1\} \in \domain$. 
\end{definition}

\begin{definition}
For a tree  $\bintree = (\domain , \ell)$ let $\subtree{w} = (\domain_w , \ell_w)$ be a \emph{subtree} of $\bintree$ rooted at $w$ defined as
$\domain_w\,=\,w^{-1}\mathcal{U}=\{v | wv \in \mathcal{U}\}$
and $\ell_w(v) = \ell(wv)$.
\end{definition}

The key idea of our approach is that the behaviour of a tree can be studied using local information at each branching point, rather than requiring global information at once. This is closer in spirit to the coalgebraic presentation of trees (see, for instance, \cite{Jacobs_2016}). One advantage of coalgebraic modelling is that it allows a unified treatment of defining conformances (like bisimilarity or bisimulation distance); for instance, we will use one of these methods (cf.\thinspace \cref{rem:coupling-basedlifting}) to define bisimulation distance over trees in \cref{sec:quantitative}. This naturally leads to the introduction of a successor function.

We first observe that the successor of a node is either in $\Sigma_0$, in the case of a leaf, or in $\Sigma_2 \times \domain \times \domain$ in the case of a branching node. To formalise this distinction, we introduce 
$F\domain= \Sigma_0 \uplus (\Sigma_2 \times \domain \times \domain)$ specifying the possible types of successors.

\begin{definition}\label{def:bintree}
Let $\bintree = (\domain , \ell)$ be a labelled tree and let $\arr :\domain \to \{0,2\}$ be the \emph{arity} function denoting the arity of a node in $\domain$ (i.e., $\arr (w) =0$ if $ \ell(w) \in \Sigma_0$; $\arr (w)=2$ otherwise). For a labelled tree $\bintree = (\domain , \ell)$, its \emph{successor function} is a map $\gamma: \domain \rightarrow F\domain$ defined as follows:
\[
\gamma (w)=
\begin{cases}
\ell(w)& \text{if}\ \arr(w)=0\\
\bigl(\ell(w),w.0 ,w.1\bigr)& \text{if}\   \arr(w)=2
\end{cases}
\]
\end{definition}

Henceforth we write $\bintree$ (and similarly $\bintree'$) to denote a labelled binary tree $(\domain, \ell)$ equipped with a successor function $\gamma$ (respectively $(\domain', \ell')$ with $\gamma'$).

\begin{remark}
Although we present our definitions for trees, all notions introduced in this paper extend naturally to any similar structure consisting of a carrier set $X$ together with a successor function $\gamma \colon X \to FX$. In several places, we will apply definitions given for trees to other structures of the same type, such as automata with carrier set $A$ and transition function $\delta \colon A \to FA$, without restating them explicitly. This slight abuse of notation should cause no confusion, as the arguments remain unchanged.

Readers familiar with coalgebraic terminology will recognise that all these structures are instances of $F$-coalgebras for a polynomial endofunctor $F$ over the category of sets. In particular, our labelled binary trees are coalgebras for the functor $F\_ = \Sigma_0 \uplus (\Sigma_2 \times \_\times \_)$.
\end{remark}

\noindent
We end this subsection with the following well-known result on substituting $\bintree'$ inside $\bintree$ at node $w\in\domain$.
\begin{proposition}
\label{prop:tree substitution}
Let $\bintree $ and $\bintree' $ be binary trees. For any node $w \in \domain$, by substituting the subtree $\subtree{w}$ with $\bintree'$ we get a new tree defined as $\bintree[w / \bintree'] = (\overline{\domain}, \overline{\ell} , \overline{\gamma})$ where 
$\overline{\domain}=
\bigl(\domain \setminus \{ v \in \domain \mid w \le v \}\bigr)
\;\cup\;
\{\, w .v \mid v \in \domain' \,\},
$
and the successor function $\overline{\gamma}$ is given by
\[
\overline{\gamma}(w)
=
\begin{cases}
\gamma(u) & \text{if } u \in \domain \setminus \{ v \mid w \le v \},\\
\sigma & \text{if } u = w.v\in \{\, w . v \mid v \in \domain' \,\}, \, \gamma'(v) = \sigma \in \Sigma_0 \\
(\sigma , u.0 , u.1)& \text{if } u = w.v\in \{\, w . v \mid v \in \domain' \,\}, \, \gamma'(v) =  (\sigma , v.0 , v.1) \text{ for }\sigma \in \Sigma_2
\end{cases}
\]
\end{proposition}
\subsection {Automata Acceptance Game}
\label {subsec: acceptance game}

When viewing programs as trees, a common way of analyzing if they satisfy a property is a two-player game characterising some form of  bisimulation. In order to define this game, we first need to ``lift'' a relation between nodes of two tree structures to a relation between their successors, as described in the following definition. We then introduce the acceptance game.

\begin{definition}\label{def:rlift}
Let $\bintree,\,\bintree'$ be two trees over the domains $\domain,\,\domain'$, respectively. For a relation $R \subseteq \domain \times \domain'$, we define the \emph{lifted relation} $\bar{R} \subseteq F\domain \times F\domain'$ as follows:
$$
\alpha \mathrel{\bar{R}} \beta \iff
\alpha = \beta \in \Sigma_0
\ \text{or }\
\left\{
\begin{aligned}
& \alpha = (\sigma, w_1, w_2)\; , \;  \beta = (\sigma, w'_1, w'_2)\; \text{ for} \;  \sigma \in \Sigma_2, \\
& \text{and }\; w_1 \mathrel{R} w'_1\; ,\text{ and } \; w_2 \mathrel{R} w'_2.
\end{aligned}
\right.
$$
\end{definition}

\begin{definition}
A relation $R \subseteq  \domain \times \domain'$ is called a \emph{bisimulation} between
$\bintree$ and $\bintree'$ iff $R \subseteq (\gamma \times \gamma')^{-1}(\bar R)$. In other words, for every $w\in \domain,\,w'\in\domain'$ we have 
\[
w \mathrel R w'\implies \Big( \ell (w) = \ell' (w') \land \big( \arr(w) = 2 \implies w.0 \mathrel{R} w'.0 \land w.1 \mathrel{R} w'.1 \big) \Big).
\]

Two subtrees $\subtree{w} $ , $\subtree{w'}'$ are \emph{bisimilar}, denoted $\subtree{w} \bisim {\subtree{w'}'}$, iff there is a bisimulation relation $R$ that includes $(w , w')\in R$. In particular, $\bintree \bisim \bintree' $ iff there is a bisimulation relation $R$ that includes $(\rootnode_{\bintree} , \rootnode_{\bintree'})\in R$.
\end{definition}

\begin{definition}
A nondeterministic tree automaton $\mathbb{A}$ is a quadruple $(A, a_0, \Delta, Acc)$, where $A$ is a finite set of states, $a_0 \in A$ is the initial state,
$\Delta : A \to \power FA$ is the transition function,
and $Acc \subseteq A^\omega$ is the infinite acceptance condition.
\end{definition}

\begin{remark}
  A more general acceptance game appears in~\cite{kupke-venema:coalg-aut-theory}, for an alternating automata with $\Delta \colon A \rightarrow \power\power FA$. Our game can be seen as an instance of this general game, in which the choice of Falsifier is always unique, as we do not model alternation. However, as shown in ~\cite[Thm.~5.2]{kupke-venema:coalg-aut-theory}, for weak pullback preserving endofunctors $F$ over sets, alternating and non-deterministic automata have the same expressive power. The functor for binary trees satisfies this condition. Hence, in our case alternating $F$-automata can be simulated by non-deterministic ones over the same functor -- this is a coalgebraic generalisation of the well-known Rabin's theorem on the complementation of tree automata~\cite{rabin69:deci}.

\end{remark}

We will now introduce the automata acceptance game, not in its full generality but for the case where all infinite runs/plays are accepted and thus won by the Verifier. As the acceptance condition $Acc=A^\omega$  we represent automata by a triple.

\begin{definition}
\label {def:acceptance game}
Let $\mathbb{A}=(A, a_0, \Delta)$ be a tree automaton and let $\bintree$ be a labelled tree. The acceptance game associated with $\mathbb{A}$ and $\subtree{w}$ is a two player game defined as follows:
\\
Each round of a game starts with a basic position $(a,w)$ where $a\in A ,\, w \in \domain$, then
\begin{enumerate}
\item At position $(a,w)$ Verifier (often denoted by $\exists$ in the sequel) can choose the next state $\delta (a) \in F A$ of the automaton from the possible choices in $\Delta (a)$.
\item Then $\exists$ picks a relation $R\subseteq A\times \domain$ such that $(\delta(a) , \gamma (w)) \in \bar{R}$.
\item At position $R$, Falsifier (denoted $\forall$) can choose any pair $(b , v)$ that belongs to $R$.
\end{enumerate}
The game then advances to the next round with basic position $(b,v)$.
\\
If a player cannot play then they lose, and over an infinite play $\exists$ wins.
We say that $\mathbb{A}$ accepts $\subtree{w}$ when $(a_0 , w) $ is a winning position for $\exists$. The set of winning positions for Verifier is denoted as $\kvwin$.
\\
Henceforth we will present games in tabular form, providing all the relevant information.
\begin{table}[H]
\centering
\begin{tabular}{|c|c|c|}
\hline
\textbf{Position} & \textbf{Player} & \textbf{Admissible moves}  \\ \hline
$(a,w) \in A \times \domain$ & $\exists$ & $\{ (\delta (a) , w)\in FA \times  \domain  \mid \delta(a) \in \Delta(a)\}$ \\ \hline
$(\delta (a) , w )\in FA\times  \domain $ & $\exists$& $R \in \power (A \times \domain)  \mid (\delta (a) , \gamma(w)) \in \bar{R}\}$ \\ \hline
$R \in \power(A \times \domain)$ & $\forall$ & $\{(b,v)\mid (b,v)\in R \}$ \\ \hline
\end{tabular}
\caption{Acceptance game for a tree automaton and a binary tree
~\cite{kupke-venema:coalg-aut-theory}.}
\label{table:accept game}
\end{table}

\end{definition}

\begin{remark}
\label{rem:smallestrelation}
The above game is the general coalgebraic form of the acceptance game. This can be simplified to its standard form by observing that if $\exists$ has a winning move $R$, then its restriction to immediate successor pairs is also winning, since the validity of $R$ depends only on successors. Concretely, any such $R$ is either $\varnothing$ (when $\ell(a)=\ell(w)\in \Sigma_0$) or $\{(b_0,w.0),(b_1,w.1)\}$ (when $\ell(a)=\ell(w)\in \Sigma_2$, $\delta(a)=(\sigma,b_0,b_1)$, and $\gamma(w)=(\sigma,w.0,w.1)$). Consequently, Falsifier’s choices may be restricted to these successors without loss of generality. Formally, a position is winning in the original game iff it is winning in this restricted version.
\end{remark}

\section {Quantitative Games for Bisimulation Distance and \texorpdfstring{$\epsilon$}{epsilon}-acceptance}
\label{sec:quantitative}

\subsection {Distance Lifting}
\label{subsec:distance lifting}

In order to obtain quantitative versions of (bisimulation/acceptance) games and their associated results, we need a `general' notion of distance lifting (instead of relation lifting). A relation can be seen as a function taking values in $\{0,1\}$; this perspective naturally extends to functions to the unit interval $[0,1]$. In this setting, we call such functions distances and introduce a suitable notion of distance lifting.

\begin{definition}\label{def:dlift}
Let $\bintree$ and $\bintree'$ be two trees of domains $\domain$ and $\domain'$, respectively. Then every order preserving function $\fun \colon [0,1]\times[0,1]\rightarrow [0,1]$ gives rise to a distance lifting $\bar d\colon F\domain \times F\domain' \to [0,1]$ for a given $d\colon \domain \times \domain' \to [0,1]$ as follows:
\[
\bar{d}(\alpha, \beta)=
\begin{cases}
0 & \text{if } \alpha= \beta \in \Sigma_0 \\
\fun\bigl(d(w_1,w'_1) \, , \, d(w_2,w'_2)\bigr)& \text{if } \alpha = (\sigma ,w_1 ,w_2 ) \; , \; \beta = (\sigma ,w'_1 ,w'_2 ) \; \text{ for} \;  \sigma \in \Sigma_2 \\
1 &  \text{otherwise}
\end{cases}.
\]
\end{definition}

\begin{remark}\label{rem:coupling-basedlifting}
Perhaps, for category theory minded readers, it is useful to note that the above distance lifting is an instance of coupling based lifting \cite{Bonchi_et_al:UptoFibration} of a set endofunctor $F\colon \set \to \set$. The following construction is similar to \cite[Equation~5]{Bonchi_et_al:UptoFibration}, but our $\V$-valued relations (where $\V$ is the Lawvere quantale $([0,1],\geq)$) are over two different sets (unlike in op.\thinspace cit.\thinspace where they are restricted to binary relations over the same set).

Now the function $\fun$ gives rise to an evaluation function $\evfun\colon F[0,1] \to [0,1]$ as follows: $\sigma \in\Sigma_0 \mapsto 0$ (for each $\sigma\in\Sigma_0$) and $(\sigma,r,r')\mapsto \fun(r,r')$ (for each $(\sigma,r,r')\in \Sigma_2 \times [0,1] \times [0,1] $). Notice that the evaluation function $\evfun$ is also monotone whenever $\fun$ is, as soon as we enriched $F[0,1]$ by lifting the order $\geq$ on $[0,1]$ and letting $\Sigma_0,\Sigma_2$ to be discretely ordered. As a result, we have a following composition of monotone maps:
\begin{equation}\label{eq:abstractdlift}
([0,1]^{\domain \times \domain'},\geq) \xrightarrow{\evfun \circ F(\_)} ([0,1]^{F(\domain \times \domain')},\geq)  \xrightarrow{(\lambda_{U,U'})_!} ([0,1]^{F(\domain) \times F(\domain')},\geq),
\end{equation}
where $\lambda_{U,U'}$ is the tupling of projections $\langle F\proj \domain, F\proj \domain'\rangle \colon F(\domain \times \domain') \to F(\domain) \times F(\domain')$ and $f_!\colon [0,1]^X \to [0,1]^Y$ (for a function $f\colon X \to Y$) is the left adjoint to reindexing map $f^*$ (see \cite{Bonchi_et_al:UptoFibration} for more details) given by $f_!(p)(y) = \inf_{f(x)=y} p(x)$.
In this categorical setting, we obtain the following result (which is not required to understand the rest of this paper, but aimed at the reader with a coalgebraic appetite):
\begin{lemma}
\label{lem:DliftIsClift}
Let $d\colon \domain \times \domain' \to [0,1]$. Then, distance lifting of $d$ is equivalent to the application of the composition given in \eqref{eq:abstractdlift} on $d$; i.e., $\bar d = (\lambda_{\domain,\domain'})_! (\evfun \circ F(d))$.
\end{lemma}
\begin{proof}
Let $\alpha\in F(\domain),\beta\in F(\domain')$. Then using \eqref{eq:abstractdlift} we have
\[(\lambda_{\domain,\domain'})_! (\evfun \circ F(d)) (\alpha,\beta)= \inf_{r\in \Gamma(\alpha,\beta)} \evfun \circ F(d) (r),\]
where $\Gamma(\alpha,\beta) = \{r \in F(\domain\times \domain') \mid F(\proj \domain)(r) = \alpha \land F(\proj {\domain'})(r)=\beta \}$. Then we identify the following cases based on the type of $\alpha\in F(\domain),\beta\in F(\domain')$:
\begin{enumerate}
\item Let $\alpha,\beta\in \Sigma_0$. Then $r\in\Gamma(\alpha,\beta) \iff (\alpha = r \land r = \beta)$. Then we distinguish two cases:
\begin{enumerate}
\item Let $\alpha=\beta$. Then, using the definition of $\evfun$ we find
\[(\lambda_{\domain,\domain'})_! (\evfun \circ F(d)) (\alpha,\beta) = \evfun \circ F(d) (\alpha) = 0 = \bar d (\alpha,\beta) \]
\item\label{case1b} Let $\alpha\neq \beta$. Then, $\Gamma(\alpha,\beta)=\emptyset$ and using empty infima corresponds to $1$ we find that
\[
(\lambda_{\domain,\domain'})_! (\evfun \circ F(d)) (\alpha,\beta) = 1 = \bar d(\alpha,\beta).
\]
\end{enumerate}
\item Let $\alpha\in \Sigma_0$ and $\beta\in \Sigma_2 \times \domain' \times \domain'$. Then $\Gamma(\alpha,\beta)=\emptyset$ and the remainder is similar to \cref{case1b}.
\item Let  $\alpha \in \Sigma_2 \times \domain \times \domain$ and $\beta \in \Sigma_0$. Similar to the above case.
\item Let $\alpha \in \Sigma_2 \times \domain \times \domain$ and $\beta \in \Sigma_2 \times \domain' \times \domain'$.  Then we further distinguish two cases (below $\proj {\Sigma_2}$ is the mapping that project elements from $\Sigma_2$):
\begin{enumerate}
\item Let $\proj {\Sigma_2} (\alpha) = \proj {\Sigma_2}(\beta)$. I.e., there are $\sigma\in\Sigma_2$, $w_1,w_2\in \domain$, and $w_1',w_2'\in\domain'$ such that $\alpha=(\sigma,w_1,w_2)$ and $\beta=(\sigma,w_1',w_2')$. Then, there is a unique $r\in \Gamma(\alpha,\beta)$ and is given by $r=(\sigma,(w_1,w_1'),(w_2,w_2'))\in F(\domain \times \domain')$. Using the definition of $\evfun$ we get:
\[
(\lambda_{\domain,\domain'})_! (\evfun \circ F(d)) (\alpha,\beta) = \evfun \circ F(d)(r) = \fun(d(w_1,w_1'), d(w_2,w_2')) = \bar d(\alpha,\beta).
\]
\item Let $\proj {\Sigma_2} (\alpha) \neq \proj {\Sigma_2}(\beta)$. Then $\gamma(\alpha,\beta)=\emptyset$ and the remainder is similar to \cref{case1b}.
\end{enumerate}
\end{enumerate}
\end{proof}

\end{remark}

\begin{definition}
A distance function $d: \domain \times \domain' \rightarrow [0,1]$ is a \emph{bisimulation distance function} iff $\forall (w,w') \in \domain \times \domain' \quad d(w,w') \geq \bar{d} (\gamma(w) , \gamma'(w'))$, where
\[
\bar{d}(\gamma(w), \gamma'(w') )=
\begin{cases}
1 & \text{if } \ell(w) \neq \ell'(w') \\
0 & \text{if } \ell(w) = \ell'(w') \land \arr(w) = 0 \\
\fun\bigl(d(w.0,w'.0) \, , \, d(w.1,w'.1)\bigr)& \text{otherwise}
\end{cases}
\]
\end{definition}

Just like traditional bisimilarity $\bisim$ is the largest bisimulation relation, we can now define the bisimulation distance $\bd$ as the greatest post fixpoint of the following functional:
\[
([0,1]^{\domain \times \domain'},\geq) \ \xrightarrow{\bar\_} \ ([0,1]^{F(\domain) \times F(\domain')},\geq) \ \xrightarrow{(\gamma \times \gamma') \circ \_} \ ([0,1]^{\domain \times \domain'},\geq).
\]
Note that the well known Knaster-Tarski fixpoint theorem guarantees the existence of $\bd$.
We say that the \emph{bisimulation distance} of $\subtree{w}$ and ${\subtree{w'}'}$ is $\varepsilon$ iff $\bd (\subtree{w} , {\subtree{w'}'})=\varepsilon$. In particular, the bisimulation distance of $\bintree$ and $\bintree'$ is $\varepsilon$ iff $\bd(\rootnode,\rootnode)=\varepsilon$.

\subsection {Bisimulation Distance, Game-Theoretically}
\label {bdistance game}

 Bisimulation distance enables one to analyse how far apart the behaviours of two trees sit. This can be formalised by introducing a two-player game (cf.\thinspace \cref{def:bdgame}) in which Verifier aims to establish that the distance between the trees is at most $\varepsilon$, while Falsifier attempts to decrease $\varepsilon$ to a point where Verifier can no longer maintain a winning strategy.

\begin{definition}\label{def:bdgame}
Let $\bintree$ and $\bintree'$ be two trees over the domains $\domain$ and $\domain'$, respectively. The \emph{bisimulation distance game} between $\bintree$ and $\bintree'$ is a two player game defined as in  \Cref{table:bisimdistance-game}.

If a player cannot play then they lose, and over an infinite play $\exists$ wins the game. If starting from a position $(w,w',\varepsilon)$ Verifier wins the game we say that $(w,w',\varepsilon)$ is a winning position for $\exists$ and we denote the set of such positions as $\bdwin$. Furthermore, if a triple $(w,w',\varepsilon)$ belongs to $\bdwin$, we say that $\varepsilon$ is a winning distance for $w,w'$.
\end{definition}
\begin{table}[H]
\centering
\begin{tabular}{|c|c|c|}
\hline
\textbf{Position} & \textbf{Player} & \textbf{Admissible moves}  \\ \hline
$(w,w',\varepsilon_1)\in \domain \times \domain'\times [0,1]$ & $\exists$ & $\{ d: \domain \times \domain' \rightarrow [0,1]  \mid \bar{d}(\gamma(w) , \gamma' (w')) \leq \varepsilon_1$ \} \\ \hline
$d \in \domain \times \domain'\rightarrow[0,1]$ & $\forall$ & $\{(v , v' , \varepsilon_2) \in \domain\times \domain' \times [0,1] \mid d(v , v') \leq \varepsilon_2\}$ \\ \hline
\end{tabular}
\caption{The bisimulation distance game for two binary trees $\bintree$ and $\bintree'$}
\label{table:bisimdistance-game}
\end{table}

\begin{proposition}
\label{prop:basic bdwins}
The bisimulation distance game has the following properties:
\begin{enumerate}
\item \label{prop:basic bdwins:1} For all $w \in \domain$ and $w' \in \domain' $, the triple $(w,w' , 1) \in \bdwin$.
\item \label{prop:basic bdwins:2} If $(w,w',\varepsilon_1) \in \bdwin$ and $\varepsilon_2 \geq \varepsilon_1$, then $(w,w' , \varepsilon_2) \in \bdwin$.
\end{enumerate}
\end{proposition}
\begin{proof}
For each part we present a winning strategy for Verifier i.e a valid move for each of their turns:
\begin{enumerate}
\item At each round, $\exists$ can choose the constant distance function $d= 1$. Since the starting position is $(w,w' , 1)$ this $d$ satisfies the property $\bar{d}(\gamma(w) , \gamma' (w')) \leq 1$ and the game moves, by the choice of $\forall$ to the next position $(v,v',1)$ and the same argument applies.

\item Assume $(w,w',\varepsilon_1) \in \bdwin$. Then $\exists$ has a winning strategy from this position, in particular a first move given by some distance $d$ such that
$$\bar{d}(\gamma(w),\gamma'(w')) \leq \varepsilon_1 \leq \varepsilon_2$$
So, $d$ is also a valid move from $(w,w',\varepsilon_2)$. Moreover, since $d$ is part of a winning strategy, every response of $\forall$ leads to a position in $\bdwin$. This means, the same strategy is winning from $(w,w',\varepsilon_2)$, and therefore $(w,w',\varepsilon_2) \in \bdwin$.
\end{enumerate}
\end{proof}

\begin{theorem}
\label {thm:win iff bdistance}
For a bisimulation distance game associated with $\bintree$ and $\bintree'$ we have
 \[
   \forall_{w\in\domain,w'\in\domain',\varepsilon\in [0,1]}\quad (w,w',\varepsilon) \in \bdwin \Longleftrightarrow \bd (\subtree{w} , {\subtree{w'}'})\leq \varepsilon .
 \]
\end{theorem}
\begin{proof}
\fbox{$\Leftarrow$} Assume that $\bd(\subtree{w},{\subtree{w'}}') \leq \varepsilon$, hence $d^*(w,w') \leq \varepsilon$.
We describe a winning strategy for Verifier: at every round, $\exists$ choose the distance $d^*$.

At the initial position $(w,w',\varepsilon)$, this is a valid move since $\varepsilon \geq d^*(w,w') \geq \bar{d}(\gamma(w),\gamma'(w'))$.
After any move of Falsifier, the game moves to a basic position $(w_1,w'_1,\varepsilon_1)$ with $d^*(w_1,w'_1)\leq \varepsilon_1$, and the same argument applies. At every position $(v,v',\kappa)$ reached during the play, the invariant $d^*(v,v') \leq \kappa$ is preserved. Hence, $\exists$ can continue playing $d^*$ indefinitely which means $\exists$ has a winning strategy starting from $(w,w',\varepsilon)$. So $(w,w',\varepsilon)\in \bdwin$.
\\

 \fbox{$\Rightarrow$} Assume $(w,w',\varepsilon) \in \bdwin$. We first construct a bisimulation distance $d$ such that $d(w,w') \leq \varepsilon$. Define
$$
d(v,v') = \inf \{ \kappa \mid (v,v',\kappa) \in \bdwin\}
\quad \text{for all } v \in \domain,\; v' \in \domain'.
$$

We show that $d$ is a bisimulation distance, i.e.
$$
d(v,v') \geq \bar{d}(\gamma(v), \gamma'(v'))
\quad \text{for all } v \in \domain,\; v' \in \domain'.
$$

Fix $v,v'$ and let $d(v,v') = \mu$. If $\mu = 1$, the inequality is trivial. Suppose $\mu < 1$. By definition of $d$ and \Cref{prop:basic bdwins}(\Cref{prop:basic bdwins:2}), for every $\kappa > \mu$ we have $(v,v',\kappa) \in \bdwin$. Hence $\exists$ has a winning strategy from $(v,v',\kappa)$, by playing a distance $d_\kappa$ such that
$
\bar{d_\kappa}(\gamma(v), \gamma'(v')) \leq \kappa,
$
and any valid choice of $\forall$, in particular all triples $(u,u',d_\kappa(u,u'))$, belong to $\bdwin$.
\\
By definition of $d$, it follows that $d(u,u') \leq d_\kappa(u,u')$ for all $u,u'$, hence $d \succeq d_\kappa$. So by monotonicity of the lifting,
$$
\bar{d}(\gamma(v), \gamma'(v')) \leq \bar{d_\kappa}(\gamma(v), \gamma'(v')) \leq \kappa.
$$
Since this holds for all $\kappa > \mu$, we obtain
$$
\bar{d}(\gamma(v), \gamma'(v')) \leq \mu = d(v,v').
$$
Therefore, $d$ is a bisimulation distance.

Finally, since $d^*$ is the best bisimulation distance, we have $d^*(w,w') \leq d(w,w')$. Moreover, $d(w,w') \leq \varepsilon$ by definition of $d$ and the assumption $(w,w',\varepsilon) \in \bdwin$. Hence
$$
\bd(\subtree{w},{\subtree{w'}}') = d^*(w,w') \leq d(w,w') \leq \varepsilon,
$$
which concludes the proof.
\end{proof}

\begin{corollary}
The set of winning distances for any $w \in \domain$ and $w' \in \domain' $ is of the form of a closed interval $[\varepsilon , 1]$ where $\varepsilon =  \bd (\subtree{w} , {\subtree{w'}}')$.
\end{corollary}

Consequently, the set of winning distances is closed, and the infimum is attained. This establishes the existence of a minimal winning distance $\varepsilon$ which coincides with the bisimulation distance between the corresponding subtrees. 

Furthermore, the bisimulation distance can be characterised in terms of winning values of $\varepsilon$: either the optimal value is obtained, showing the exact distance from a property, or we have an upper bound given by a winning $\varepsilon$. Examples of such optimal winning values and their interpretation are discussed in \Cref{sec:optimal epsilon}. 

\subsection {The \texorpdfstring{$\epsilon$}{epsilon}-acceptance Game}

\begin{definition}
Let $\mathbb{A}=(A, a_0, \Delta)$ be a tree automaton and $\bintree$ a binary tree over a domain $\domain$ with successor function $\gamma\colon \domain\to F(\domain)$. The acceptance game  associated with $\mathbb{A}$ and $\subtree{w}$ is a two player game where the moves of each player are shown in \Cref{tab:eaccept game}.
\begin{table}[H]
\centering
\begin{tabular}{|c|c|c|}
\hline
\textbf{Position} & \textbf{Player} & \textbf{Admissible moves}  \\ \hline
$(a,w,\varepsilon_1)$ & $\exists$ & $\{ (\delta (a),w,\varepsilon_1)\in F(A) \times \domain \times [0,1] \mid \delta (a) \in \Delta(a) \}$ \\ \hline
$ (\delta (a),w,\varepsilon_1)\in F(A) \times \domain \times [0,1]$ & $\exists$ & $\{ d: A \times \domain \rightarrow [0,1]  \mid \bar{d}(\delta (a) , \gamma(w)) \leq \varepsilon_1$ \} \\ \hline
$d \in A \times \domain\rightarrow[0,1]$ & $\forall$ & $\{(b , v, \varepsilon_2) \in A \times \domain \times [0,1] \mid d(b , v) \leq \varepsilon_2\}$ \\ \hline
\end{tabular}
\caption{$\epsilon$-acceptance game for a tree automaton $\mathbb{A}$ and a binary tree $\bintree$}
\label{tab:eaccept game}
\end{table}
\noindent
If a player cannot play they lose, and for any infinite play $\exists$ wins.
We denote by $\ekvwin$ the set of winning positions for Verifier and $\subtree{w}$ is said to be $\varepsilon$-accepted by $\mathbb{A}$ when $(a_0, w, \varepsilon) \in \ekvwin$.
\end{definition}

\begin{remark}
\label{rem:min distance}
By a similar argument to the bisimulation distance game, if $(a,w,\varepsilon_1) \in \ekvwin$ and $\varepsilon_1 \le \varepsilon_2$, then $(a,w,\varepsilon_2) \in \ekvwin$. Hence, without loss of generality, we may assume that choices of $\forall$ are of the form $(b,v,d(b,v))$, where $d$ is the function previously chosen by $\exists$.
\end{remark}

\section {\texorpdfstring{$\epsilon$}{epsilon}-Acceptance Is Acceptance Plus Bisimulation Distance}
\label{sec:main thm}
We are now ready to state and prove the main theorem of the paper, establishing the connection between acceptance and quantitative acceptance.

Throughout this section, we work with a tree automaton $\mathbb{A}=(A, a_0, \Delta)$ where every infinite word is accepted and $\Delta(a) \neq \emptyset$ for every $a\in A$ (i.e. the automaton is deadlock free).

\begin{theorem}
\label{thm:main}
Let $\mathbb{A}=(A, a_0, \Delta)$ be a deadlock free tree automaton and let $\bintree$ be a labelled tree. Then the following statements are equivalent:
\begin{enumerate}
\item\label{C1:mainthm} $\bintree$ is $\varepsilon_0$-accepted by $\mathbb{A}$.
\item\label{C2:mainthm} There is a binary tree $\bintree'$ such that $\bd( \bintree' , \bintree)\leq \varepsilon_0$, and $\bintree'$ is accepted by $\mathbb{A}$.
\end{enumerate}
\end{theorem}

We prove the two directions of the theorem separately; the proof of the direction $2\implies 1$ is immediate once we have the following lemma which states that the positions of the form $(a,w,1)$ are always winning for Verifier in the $\epsilon$-acceptance game.


\begin{lemma}
\label{lem:1 is accepted}
For a deadlock free automaton $\mathbb A$ and a tree $\bintree$,  we have $(a,w,1)\in\ekvwin$ (for every $a\in A,w\in\domain$).
\end{lemma}

\begin{proof}
At position $(a, w, 1)$, $\exists$ can choose any $\delta(a) \in \Delta(a)$ since $\Delta(a) \neq \emptyset$ and distance function $d$ is the constant $1$ function. Then any choice of $\forall$ will be of the form $(b, v, 1)$ that will be the new position of the game and the same argument applies. Therefore $\exists$ always has a move, and the play will be infinite. So $(a, w, 1) \in \ekvwin$.
\end{proof}

\begin{proof}[Proof of \cref{thm:main} ($2\implies 1$)]
Assume the existence of such $\bintree'$ as given in \cref{C2:mainthm} of \cref{thm:main}. We prove a stronger property in the sense:
\[
\forall a\in A,w\in \domain,w'\in\domain'\ \left( (a,w')\in\kvwin \land (w',w,\varepsilon)\in\bdwin \right) \implies (a,w,\varepsilon)\in\ekvwin.
\]
So let $a,w,w'$ be such that $(a,w')\in\kvwin$ and $(w',w,\varepsilon)\in\bdwin$ hold. Without loss of generality, let $\varepsilon<1$; otherwise, the result follows from \cref{lem:1 is accepted}. Since $(a,w')\in\kvwin$ Verifier plays some $\delta(a)\in\Delta(a)$ and relation $R\subseteq A\times\domain'$ such that $\delta(a) \mathrel {\bar R} \gamma'(w')$. Moreover, since $(w',w,\varepsilon)\in\bdwin$ we know that Verifier in the bisimulation distance game plays some $d\colon \domain' \times\domain \to [0,1]$ such that $\bar d(\gamma'(w'),\gamma(w)) \leq \varepsilon$. So consider the function
\[
d^R (b,v) = \inf_{v'\in \domain', b \mathrel R v'} d(v',v) \quad \text{(for each $b\in A,v\in\domain$)}.
\]
We show that $d^R$ is a valid move in the $\epsilon$-acceptance game. Note that $\ell(a)=\ell(w')$ (resp.\thinspace $\ell(w') = \ell(w)$) because $\delta(a) \mathrel {\bar R} \gamma'(w')$ (resp.\thinspace $\bar d(\gamma'(w'),\gamma(w))\leq \varepsilon<1$). Thus, $\ell(a)=\ell(w)$ and performing the case distinction on the type of $\ell(w)$ we get the following two cases:
\begin{enumerate}
\item Let $\ell(w)\in \Sigma_0$. Then by \cref{rem:smallestrelation} we know that $R=\emptyset$; therefore, $d^R$ corresponds to constant $1$ function. Moreover, from \cref{def:dlift}, we know that $\overline{d^R}(\delta(a),\gamma(w))=0\leq \varepsilon$; thus $d^R$ is a valid move. Furthermore, any response of Falsifier will be of the form $(b,v,1)$ which belongs to $\ekvwin$ by \cref{lem:1 is accepted}. So $(a,w,\varepsilon)\in\ekvwin$.
\item Let $\ell(w)=\sigma\in \Sigma_2$. Then, $\delta(a)=(\sigma,b_0,b_1)$ (for some $b_0,b_1\in A$), $\gamma(w)=(\sigma,w.0,w.1)$, and $\gamma'(w')=(\sigma,w'.0,w'.1)$. By \cref{rem:smallestrelation} we know that $R=\{(b_0,w'.0),(b_1,w'.1)\}$. Then $d^R(b_i,w.i) = d(w'.i,w.i)$ (for $i\in\{0,1\}$). Moreover, using \cref{def:dlift} we get:
\begin{align*}
\overline{d^R}(\delta (a),\gamma(w)) = &\ \fun(d^R(b_0,w.0),d^R(b_1,w.1) \\
=&\ \fun(d(w_0',w_0),d(w_1',w_1)))\\
 =&\  \bar d(\gamma'(w'),\gamma(w)) \leq \varepsilon.
\end{align*}
Thus, $d^R$ is a valid move in this case. Furthermore, any response of Falsifier in the $\epsilon$-acceptance game will be either of the form $(b,v,1)$ (which is winning for $\exists$ by \cref{lem:1 is accepted}), or is of the form $(b_i,w.i,\varepsilon')$ (for $\varepsilon'\geq d^R(b_i,w.i)=d(w'.i,w.i)$ and $i\in\{0,1\}$).

So it remains to show that $(b_i,w.i,\varepsilon') \in \ekvwin$. Suppose Falsifier in the $\epsilon$-acceptance game plays with $(b_i,w.i,\varepsilon')$ (for some $i\in\{0,1\}$ and $\varepsilon'\geq d(w'.i,w.i)$). Then $(b_i, w'.i) \in \kvwin$ as it can be the response of Falsifier in the acceptance game to the winning move $R$ (cf.\thinspace \cref{rem:smallestrelation}). In addition, $(w'.i, w.i, \varepsilon) \in \bdwin$ as it is a valid response for Falsifier in the bisimulation distance game.  Thus, we are back to the same condition as the start of the round and so $\exists$ can repeat the same argument.\qedhere
\end{enumerate}
\end{proof}

The converse direction (i.e.``$1 \implies 2$'') requires more work since we need to construct a labelled tree $\bintree'$ which itself is based on the idea of a run tree. Suppose $(a,w,\varepsilon)\in\ekvwin$, then a run tree $\bintree_{(a,w,\varepsilon)}$ is simply the tree generated by following a winning strategy of Verifier (see \cref{prop:domainRunTree} for the formal definition of domain of a run tree). Now if we encountered a position $(b,v,1)$ while exploring the run tree $\bintree_{(a,w,\varepsilon)}$, the successor function functions associated to $b$ and $v$ can be of different type due to the definition of distance lifting (\cref{def:dlift}). Moreover, since these positions $(b,v,1)$ are winning for Verifier in the $\epsilon$-acceptance game, thus we cannot use the successor function of $v$ to construct our $\bintree'$. 

To this end, we replace the first occurrence of such positions $(b,v,1)$ in the run tree by the tree $\bintree_{(b)}$ accepted by the automaton state $b$ (cf.\thinspace \cref{lem:existence of accepted tree}). In other words, the tree $\bintree'$ (used in the proof of ``$1\implies 2$'') is the tree obtained by removing nodes of the form $(b,v,1)$ from a run tree $\bintree_{(a,w,\varepsilon)}$ and then attaching the tree $\bintree_{(b)}$ as defined in the next lemma. 
 $(a, v, 1)$ which is winning for $\exists$ by Lemma 1, or is of the form $(b_i, w.i, \varepsilon)$ for $\varepsilon \geq d(b_i, w.i), i \in \{0, 1\}$.

\begin{lemma}
\label{lem:existence of accepted tree}
For any $a \in A$, there exists a binary tree $\bintree_{(a)}$ such that $(a , \rootnode) \in \kvwin$.
\end{lemma}
\begin{proof}
The idea of the proof is to choose a fix a family of successor functions $\delta (a) \in \Delta(a)$ since $\Delta(a)$ (for each $a\in A$) is nonempty.
Now inductively define $S_n$ as follows:
\[S_0= \{ (a, \rootnode) \} \qquad S_{j+1} = \{ (c, w.0), (c', w.1) \mid \exists b\in A \quad (b, w) \in S_k \land \delta(b) = (\sigma, c, c') \}. \]
It is easy to show by induction that these sets has the property that any $w$ can appear at most once in a pair $(b, w) \in S_{|w|}$ (the notation $|w|$ denotes the length of the word $w$).
Let $S = \bigcup_{j=0}^{\infty} S_j$. Now we define $\bintree_{(a)}$ as follows:
\[ \gamma_a\colon \domain_{(a)} \to F\domain_{(a)}   \quad \text{with }\gamma_a(w)=
\begin{cases}
\delta(b) & \text{if } (b,w) \in S_{|w|} \land \delta (b)\in \Sigma_0\\
(\sigma, w.0, w.1) & \text{if } (b,w) \in S_{|w|} \land \delta (b)=(\sigma,c,c')
\end{cases},\]
where $\domain_{(a)} = \pi_2 (S)$ where $\pi_2(S)$ projects the second component from the elements in $S$. Notice that this $\gamma$ is well defined as in the second case $(c, w.0), (c', w.1)$ belong to $S$ and so $w.0, w.1$ belong to $\domain_{(a)}$. We next show that $(a, \rootnode) \in \kvwin$. At position $(a, \rootnode) \in S$, let $\exists$ choose the same $\delta(a) \in \Delta(a)$ and to play the correct $R\subseteq A \times \domain_{(a)}$ we identify two cases:
\begin{enumerate}
    \item Let $\delta(a) = \sigma \in \Sigma_0$. Then since $(a, \rootnode) \in S \implies \rootnode \in \domain_{(a)}$ and so $\gamma_a(\rootnode) = \sigma$. Therefore $\exists$ can play $R = \varnothing$ that satisfies $\delta(a) \mathrel{ \bar R} \gamma(\rootnode)$ and $\forall$ gets stuck and loses.
    \item Let $\delta(a) = (\sigma, b, b')$. Then $\gamma_a(\rootnode) = (\sigma, 0, 1)$ (recall $\rootnode$ is the empty string, so $\rootnode.i=i$ for $i\in\{
        0,1\}$). Now $\exists$ can play $R = \{ (b, 0), (b', 1) \}$ which satisfies $\delta(a) \mathrel {\bar{R}} \gamma_a(\rootnode)$ and $\forall$ can choose either $(b, 0)$ or $(b', 1)$. However, by the definition of $S$, both of these pairs belong to $S$ and so we are back to the same condition as the begining of the round. Hence $\exists$ can repeat the same argument.
\end{enumerate}
So either $\forall$ looses in some round, or the play is infinite, and therefore won by $\exists$.\qedhere
\end{proof}

Let $\bintree$ be a labelled tree such that $(a_0 , \rootnode , \varepsilon_0) \in \ekvwin$. For all of the triples $(a,w,\varepsilon) \in \ekvwin$, Verifier has a non empty set of winning moves to play which consists of a choice of successor functions in $\Delta(a)$ and a distance function. For each $(a,w,\varepsilon) \in \ekvwin$, let $\delta_{(a,w,\varepsilon)} (a) \in \Delta (a)$ and  $d_{(a,w,\varepsilon)} \colon A \times \domain \to [0,1]$ denote the successor and distance functions played by Verifier to win at the position $(a,w,\varepsilon)$. Then we let $M=\bigcup_{j=0}^{\infty} M_j$ and define $M_n$ inductively as follows: $M_0=\{(a_0,\rootnode,\varepsilon_0)\}$ and
\begin{multline*}
M_{n+1} =\Big\{ \Big(b_i, w.i, d_{(a,w,\varepsilon)}(b_i, w.i) \Big) \mid i \in \{0, 1\}  \land  (a, w, \varepsilon) \in M_n \ \land\\ \delta_{(a,w,\varepsilon)} (a) = (\sigma, b_0, b_1) \land \gamma(w) = (\sigma, w.0, w.1) \Big\}.
\end{multline*}
The next proposition states that projecting the second component from the elements in $M$ results in the domain of a run tree, denoted $\bintree_{(a_0,\rootnode,\varepsilon)}$, whenever $(a_0,\rootnode,\varepsilon)\in\ekvwin$.

\begin{proposition}\label{prop:domainRunTree}
The set $\pi_2 (M)$ (i.e. projecting the second component from the triples in $M$) is both prefix closed and sibling closed.
\end{proposition}


\begin{proof}[Proof of \cref{thm:main}($1\implies 2$)]
Assume $\bintree$ is $\varepsilon_0$-accepted by $\mathcal{A}$, so $(a_0, \rootnode, \varepsilon_0) \in \ekvwin$. Define the following sets using  the domain $M$ of a run tree constructed above:
\begin{itemize}
    \item $\domain_1 = \{ w \mid \exists a \in A, \varepsilon\in[0,1)\quad  (a, w, \varepsilon) \in M \}$;
    \item $\domain_2 = \Big\{ v \mid \exists a \in A\ \Big( (a, v, 1) \in M \land \forall u\ (u \leq v \implies \nexists b\ (b, u, 1) \in M) \Big) \Big\}$;
     \item $\domain_3 = \{ v.t \mid v \in \domain_2  \text{ with } (a, v, 1) \in M \land t \in \domain_{(a)} \}$.
\end{itemize}
Notice that $\domain_2 \subset \domain_3$ as $t \in  \domain_{(a)}  $ can be the empty string.

Now we will construct a tree $\bintree'$ with the desired property by first defining its domain as $\domain' = \domain_1 \cup \domain_3$. It is prefix-closed and sibling-closed by tree substitution (\Cref{prop:tree substitution}), since both $M$ and all $\bintree_{(a)}$ are binary trees. Moreover, define the successor function $\gamma' \colon \domain' \to F(\domain')$:
\[
\gamma' (w) = \begin{cases}
\gamma (w) & \text{if}\ w\in \domain_1\\
\gamma_a(t) & \text{if}\ w\in \domain_3 \land w=v.t \ \text{(for some $t$)}
\end{cases},
\]
where $\gamma_a$ is the successor function as defined in the proof of \Cref{lem:existence of accepted tree}.
Next we claim the following two properties hold from which the final result follows directly since the starting position $(a_0, \rootnode, \varepsilon_0)$ has the said property and therefore $(\rootnode , \rootnode , \varepsilon_0) \in \bdwin$ and $(a_0 , \rootnode) \in \kvwin$.
\begin{claim}
\begin{enumerate}
\item $\forall w \in \domain_1\cup \domain_2\quad (a, w, \varepsilon) \in \ekvwin \implies (w, w, \varepsilon) \in \bdwin$.
\item $\forall w \in \domain_1\cup \domain_2\quad (a, w, \varepsilon) \in M \implies (a, w) \in \kvwin$.\qedhere
\end{enumerate}
\end{claim}
\begin{claimproof}
For Item 1, without loss of generality, let $w\in \domain_1$ (otherwise the result follows from \Cref{prop:basic bdwins}(\Cref{prop:basic bdwins:1}). Then $(a,w,\varepsilon) \in M$ (for some $a$ and $\varepsilon<1$). Therefore $\gamma'(w) = \gamma(w) = \sigma$. If $\sigma \in \Sigma_0$, $\exists$ can choose $d= 1$ and the position is winning due to \cref{prop:basic bdwins}(\cref{prop:basic bdwins:1}). If $\sigma \in \Sigma_2$, then $\gamma'(w) = (\sigma, 0, 1) = \gamma(w)$ and $\exists$ can choose $d: \domain \times \domain' \to [0, 1]$ such that $d(w.0, w.0) = d_{(a, w, \varepsilon)}(w.0,w.0)$ and $d(w.1, w.1) = d_{(a, w, \varepsilon)}(w.1,w.1)$ and $d=1$ otherwise. Then the choice of $\forall$ is either $(v, v', 1)$ that is winning for $\exists$ by \Cref{prop:basic bdwins}(\Cref{prop:basic bdwins:1}), or it is $(w.i, w.i, d_{(a, w, \varepsilon)}(w.i,w.i))$ by \Cref{rem:min distance}. But this position has the property that it is in $\ekvwin$ and $w.i \in \domain_1\cup \domain_2$ by definition. So we are back at the same condition and $\exists$ can repeat the same strategy. So the play is infinite and therefore winning by $\exists$.

For Item 2, if $w \in \domain_2$, then by definition $\bintree'_w$ is exactly the tree that is accepted by $\mathcal{A}$ starting from $w$. So let $w \in \domain_1$. Then $(a, w, \varepsilon) \in M$ for $\varepsilon < 1$ and  $\gamma'(w) = \gamma(w) = \sigma$. If $\sigma \in \Sigma_0$, $\exists$ can choose $R = \varnothing$ and wins. If $\sigma \in \Sigma_2$, then $\gamma'(w) = \gamma(w) = (\sigma, w.0, w.1)$ and $\exists$ can choose $\delta_{a, w, \varepsilon}(a)$.
 Also as $\varepsilon < 1$  we have the property $\bar {d}_{a, w, \varepsilon} \bigl(\delta_{a, w, \varepsilon}(a) , \gamma'(w) \bigr)\leq \varepsilon < 1$ so $\delta_{a, w, \varepsilon}(a)$ has to be of the form $(\sigma , b_0 , b_1)$. Now for the relation $\exists$ can choose $R = \{ (b_0, w.0), (b_1, w.1) \}$ and the next move $(b_i, w.i)$ of $\forall$ is such that $w.i \in \domain_1 \cup \domain_2$ as $(b_i, w.i, d_{(a,w,\varepsilon)}(b_i, w.i)) \in M$. So back to the same condition and $\exists$ can repeat the same strategy. So the play is infinite and therefore winning by $\exists$.
\end{claimproof}
\end{proof}

\section {Relating Defect and Measure}
\label {sec:optimal epsilon}

In this section, we apply our main result to two case studies: measuring termination and failed executions of programs, modelled as trees. To this end, we define a measure on subtrees and relate the quantity of interest to a distance $\varepsilon$ accepted in the $\epsilon$-acceptance game. The corresponding characterisations are stated in \Cref{prop:Measuring termination} and \Cref{prop:Measuring failed executions}.

Let $\bintree$ be a binary tree. For every subtree $\subtree{w}$, the \emph{measure} of $\subtree{w}$ is defined as
$$\mu (\subtree{w})=\ 2^{-\lvert w \rvert}$$
The existence and unicity of $\mu$ are guaranteed by Carathéodory's extension theorem~\cite{bauer2001measure}. Then $(\bintree,\mu)$ is a measure space. Note that $\mu$ is also a probability measure because $\mu(\bintree) = 1$. We  also  introduce a notation reminiscent of conditional probabilities. Consider two binary strings $w,v \in \domain$, the \emph{relative measure} $\mu(v \mid w)= 2^{\lvert w \rvert -\lvert v \rvert}$ if $w$ is a prefix of $v$; $0$ otherwise.

\subparagraph{Measuring termination} Let $\Sigma_0 = \{\star\} , \Sigma_2 =\{\sigma\}$. Let $\mathbb{A}$ be a tree automaton with a single state $A=\{a_0\}$ such that $\arr(a_0)=2$, and transition function $\Delta(a_0)=\{(\sigma,a_0,a_0)\}$. Moreover, let $Acc = A^\omega$. Consider the distance lifting induced by $\fun \colon [0,1]\times[0,1]\to[0,1]$ given by $\fun(x,y)=\tfrac{1}{2}x+\tfrac{1}{2}y$. Then the following proposition holds.
\begin{proposition}
\label{prop:Measuring termination}
For a  tree $\bintree$ and  the set $L=\{l\in \domain \mid \arr(v)=0\}$ of leaves we have
\[
(a_0,\rootnode,\varepsilon)\in \ekvwin
\;\Longleftrightarrow\;
\varepsilon \ge \sum_{l\in L} \mu(l\mid \rootnode).
\]
\end{proposition}

\begin{proof}
\fbox{$\Leftarrow$} Assume $\varepsilon \ge \sum_{l \in L} \mu(l \mid \rootnode)$; starting from the position $(a_0, \rootnode, \varepsilon)$, we show that $\exists$ has a winning strategy. The choice of $\delta(a_0)$ is unique, and for the distance function $\exists$ can choose $d$ as follows. If $\ell(w)=\star$, let $d = 1$, and if $\ell(w)=\sigma$, define
\[
d(a_0,\rootnode.i)=\sum_{l \in L}\mu(l \mid \rootnode.i)\ (i\in\{0,1\}),
\qquad d(a_0,v)=1 \text{ for } v \notin \{\rootnode.0,\rootnode.1\}.
\]
In the first case, $\bar{d}(\delta(a_0),\gamma(\rootnode))=1$ because $\arr(\rootnode)\neq \arr(a_0)$. However, since $w$ is a leaf, $\sum_{l \in L}\mu(l \mid \rootnode)=1$, hence $\varepsilon \ge 1$, and $d$ is valid. Any response $(a_0,v,1)$ of $\forall$ belongs to $\ekvwin$ by \Cref{lem:1 is accepted}.

\noindent
In the second case, when $\ell(\rootnode)= \ell(a_0)=\sigma \in \Sigma_2$, we simplify $\bar{d}(\delta(a_0),\gamma(\rootnode))$ as follows:
\[
\fun \bigl(d(a_0,\rootnode.0),d(a_0 , \rootnode.1)\bigr) =
 \tfrac{1}{2}\bigl(d(a_0,\rootnode.0)+d(a_0,\rootnode.1)\bigr)
= \sum_{l \in L}\mu(l \mid \rootnode)
\le \varepsilon.
\]
The last equality follows from the definition of $\mu$. Hence $d$ is valid. Any response of $\forall$ is then either $(a_0,v,1)$ which belongs to $\ekvwin$, or $(a_0,w.i,\varepsilon_i)$ with $\varepsilon_i \geq d(a_0,.i)=\sum_{l \in L}\mu(l \mid \rootnode.i)$ for $i\in\{0,1\}$. However, in the second case, the invariant is preserved, and we return to the same condition as before (with another $w$ instead of $\rootnode$), so $\exists$ can repeat this strategy. Therefore, the play is infinite and belongs to $Acc$, and hence $(a_0,\rootnode,\varepsilon)\in \ekvwin$.

\fbox{$\Rightarrow$} To prove this direction, we need the following claim.
\begin{claim}
The only tree accepted by $\mathbb{A}$ is the full binary tree $\bintree_A$ in which every node is labelled by $\sigma$. I.e., $\bintree_A = (\domain_A,\ell)$ where $\domain_A = \{0,1\}^*$ and $\gamma_A(w) = (\sigma, w.0, w.1)$ for  $w\in\domain_A$.
\end{claim}
\begin{claimproof}
The accepted tree exists since $\mathbb{A}$ satisfies the conditions of \Cref{lem:existence of accepted tree}. Let $\bintree_A$ be such a tree. We show that every node of $\bintree_A$ must have label $\sigma$ and exactly two children.

Assume $(a_0,w)\in \kvwin$. Then $\exists$ has a move $R$ such that $\delta(a_0)\,\bar{R}\,\gamma_A(w)$, which implies $\gamma_A(w) = (\sigma, w.0, w.1)$. Moreover, any possible response of $\forall$, namely $(a_0,w.0)$ or $(a_0,w.1)$, again belongs to $\kvwin$, so the argument can be repeated.

Since $(a_0,\rootnode)\in \kvwin$, starting from the root we inductively obtain that every node is labelled by $\sigma$ and has two successors. Hence $\bintree_A$ is the full binary tree.
\end{claimproof}

Assume $(a_0,\rootnode,\varepsilon)\in \ekvwin$. By \Cref{thm:win iff bdistance}, there exists a tree $\bintree'$ accepted by $\mathbb{A}$ such that $\bd(\bintree',\bintree)\le \varepsilon$. By the above claim, $\bintree'$ is the full binary tree $\bintree_A$. Since $\gamma_A(v) = (\sigma , v.0 , v.1)$ for all $v$ and $\bd$ is a bisimulation distance, we can write (below $w\in\domain$):
\[
\bd(w,w)\geq
\overline{\bd}(\gamma_A(w) , \gamma(w)) =
\begin{cases}
1 & \text{if } \gamma(w)=\star,\\
\frac{1}{2}\bigl(\bd(w.0,w.0)+\bd(w.1,w.1)\bigr) & \text{otherwise.}
\end{cases}
\]
Hence, $\bd(w,w)=1 \text{ if } \gamma(w)=\star$ and $\bd(w,w)\geq \tfrac{1}{2}\bigl(\bd(w.0,w.0)+\bd(w.1,w.1)\bigr) \text{ otherwise.}$

\noindent
On the other hand, by definition of the measure $\mu$, for all $w \in \domain$ we have
\[
\sum_{l\in L}\mu(l\mid w)=
\begin{cases}
1 & \text{if } \gamma(w)=\star,\\
\frac{1}{2}\Bigl(\sum_{l\in L}\mu(l\mid w.0)+\sum_{l\in L}\mu(l\mid w.1)\Bigr) & \text{otherwise.}
\end{cases}
\]
Thus, $\sum_{l\in L}\mu(l\mid w) \leq \bd(w,w)$ (for $w\in \domain$); so
$
\sum_{l\in L}\mu(l\mid \rootnode)\leq \bd(\rootnode,\rootnode)\leq \varepsilon.
$
\end{proof}

\subparagraph{Measuring failed executions}
We now illustrate the approach on a very simple yet crucial property, checked by a non-deterministic automaton: the absence of error labels in a binary tree with leaves. The quantitative aspect game will \emph{tolerate} error labels, but up to a certain measure that can not be greater than  $\varepsilon$. So let $\Sigma_0 = \{\star\} , \Sigma_2 =\{\ok,\error\}$.

First, we define $\subseterr$ as the set of subtrees of $\bintree$ whose root is labeled with $\error$, and $\subseterrprime$ as the subset of those trees whose ancestors are all labeled with $\ok$. Since any element of $\subseterr \setminus \subseterrprime$ must be a subtree of an element of $\subseterrprime$, they have the same measure. Similarly, we define the conditional sets $\subseterr \mid w$ and $\subseterrprime \mid w$ where the measure originates from the
node $w$ and we have $\mu(\subseterr \mid w) = \frac{\mu(\subseterr)}{2^{\lvert w\rvert}}$ (and similarly for $\subseterrprime \mid w$), effectively normalizing the measure relative to the depth of the prefix $w$.

Let $\mathbb{A}$ be a tree automaton with a single state $A=\{a_0\}$ and the transition function $\Delta(a_0)=\{ \star , (\ok,a_0,a_0)\}$. Moreover, let $Acc = A^\omega$. Consider the same distance lifting induced by $\fun \colon [0,1]\times[0,1]\to[0,1]$ given by $\fun(x,y)=\tfrac{1}{2}x+\tfrac{1}{2}y$. Then

\begin{proposition}
\label{prop:Measuring failed executions}
$
(a_0,\rootnode,\varepsilon) \in \ekvwin\Longleftrightarrow \varepsilon \geq \mu\left(\subseterrprime \right)
$
\end{proposition}
\begin{proof}
\fbox{$\Leftarrow$} Assume $\varepsilon \geq \mu\left(\subseterrprime \right) = \mu\left(\subseterrprime \mid \rootnode \right) $. We now define a winning strategy for $\exists$ starting from $(a_0,\rootnode,\varepsilon)$. We proceed by case distinction:
\begin{enumerate}
\item If $\gamma(\rootnode) = \star$, then $\exists$ can play $\delta(a_0) = \star$ and $d=1$, which is a valid move since $\bar{d}(\delta(a_0), \gamma(\rootnode)) = 0$. And the next move of $\forall$ will be $(a_0, v, 1)\in\ekvwin$ by \Cref{lem:1 is accepted}.

    \item If $\gamma(\rootnode) = (\error, \rootnode.0, \rootnode.1)$, then $\exists$ can choose any $\delta(a_0) \in \Delta(a_0)$ and $d = 1$. In this case, $\subtree{\rootnode}$ is itself a subtree starting with an $\error$ root, so that $\mu\left(\subseterrprime \right) = 1$, and therefore $\varepsilon \geq 1$ (implying $\varepsilon = 1$). Thus $d$ is a valid choice as $\bar{d}(\delta(a_0), \gamma(\rootnode)) \leq \varepsilon$. Then any choice of $\forall$ will be of the form $(a_0, v, 1)$, that again belongs to $\ekvwin$ by \Cref{lem:1 is accepted}.
    
    \item If $\gamma(\rootnode) = (\ok, \rootnode.0, \rootnode.1)$, then $\exists$ can choose $\delta(a_0) = (ok, a_0, a_0)$ and $d$ as follows:
    \[
d(a_0, \rootnode.i) =\mu\left(\subseterrprime \mid \rootnode.i \right) \ ( i \in \{0, 1\})
    \qquad d(a_0, v) = 1 \quad \text{for } v \notin \{\rootnode.0, \rootnode.1\}
\]
    It is a valid choice because of the following derivation.
 \[
 \bar{d}(\delta(a_0), \gamma(\rootnode)) =    \frac{1}{2}\left(\mu\left(\subseterrprime \mid \rootnode.0 \right) +\mu\left(\subseterrprime \mid \rootnode.1 \right)\right)
       =\mu\left(\subseterrprime \right) \leq \varepsilon
\]

Notice that the last equality holds because the label of $\rootnode$ itself is $\ok$.  So, when measuring the set $\subseterrprime \mid \rootnode$, every $\error$-rooted subtree in that set must start from $\rootnode.0$ or $\rootnode.1$, and cannot be from $\rootnode$ itself. The next move of $\forall$ must then either be $(a_0,v,1)$ (which again belongs to $\ekvwin$), or $(a_0,\rootnode.i,\varepsilon_i)$ with $\varepsilon_i \geq d(a_0,\rootnode.i)=\mu\left(\subseterrprime \mid \rootnode.i \right)$ for $i\in\{0,1\}$. In the latter case, the invariant is preserved (with $w$ instead of $\rootnode$), and we return to the same condition as the beginning of the round, so that $\exists$ can repeat this strategy. This induces an infinite play, thus belonging to $Acc$. 
\end{enumerate}

\fbox{$\Rightarrow$} Assume that $(a_0, \rootnode, \varepsilon) \in \ekvwin$. Then, by \Cref{thm:main}, there is a tree $\bintree'$ that is accepted by $\mathcal{A}$ with $\bd(\bintree', \bintree) \leq \varepsilon$. The construction of $\bintree'$ in the proof of \cref{thm:main} has the following form:
starting from the root, $\bintree'$ coincides with $\bintree$ until an error node is met. Once we reach an error node, we replace it with a subtree accepted by $\mathcal{A}$. 
That is because the run tree following a winning strategy of $\exists$ will be equal to $\bintree$ on a path that is free of $\error$, and when we see an error node $v$, for any choice of $d$ and $\delta$ of $\exists$, $\bar{d}(\delta(a_0), \gamma(v)) = 1$, so $(a_0 , v , 1)$ will appear for the first time in the run tree and that is when we substitute $\bintree_{(a_0)}$. Moreover, $\bintree_{(a_0)}$ can be a single leaf $\star$, as it is always accepted by $\mathcal{A}$ ($\exists$ can choose $\delta (a_0) = \star$ and $d= 1$). So $\bintree'$ would be equal to $\bintree$ until it reaches an error node, at which point it becomes a leaf $\star$.

As $\bd$ is a bisimulation distance, for each $v \in \domain'$:
\[
\bd(v,v) \geq \bar{\bd}(\gamma'(v), \gamma(v)) = 
\begin{cases} 
0 & \text{if } \gamma'(v) = \gamma(v) = \star \\
\frac{1}{2}(\bd(v.0, v.0) +\bd(v.1, v.1)) &  \text{if } \gamma'(v) =\gamma(v) = (\ok, v.0 , v.1) \\
 1 & \text{if } \gamma'(v) = \star, \gamma(v) = (\error, v.0 , v.1)
\end{cases}
\]
Note that these are the only cases since by construction of $\bintree'$, it does not have any node labelled $\error$, and all of its labels $\ok$ match with the same node in $\bintree'$. 

Observe that if a tree's root is labelled with $\ok$ then it does not belong to the set of its subtrees whose root is labelled with $\error$.
From the definition of $\mu$, one then obtains that:
\[
\mu(\subseterrprime \mid v) = 
\begin{cases} 
0 & \text{if }  \gamma(v) = \star (=\gamma'(v))\\
\frac{1}{2}(\mu\left(\subseterrprime \mid v.0 \right) + \mu\left(\subseterrprime \mid v.1 \right)) &  \text{if } \gamma(v) = (\ok, v.0 , v.1) (=\gamma'(v))\\
 1 & \text{if } \gamma(v) = (\error, v.0 , v.1)\, (\text{then }  \gamma'(v) = \star)
\end{cases}
\]

It follows that $\mu\left(\subseterrprime \mid v \right) \leq \bd(v,v)$ for all $v\in \domain'$. In particular, since $\rootnode \in \domain'$,
$
\mu\left(\subseterrprime  \right)\leq \bd(\rootnode,\rootnode)\leq \varepsilon,
$
which completes the proof.
\end{proof}

\section{Conclusion}

In this paper, we extended the acceptance game \cite{kupke-venema:coalg-aut-theory} between a tree automaton and a labelled binary tree to a quantitative setting that lead to a more `robust' notion of acceptance (instead of a Boolean one). Intuitively, our main result (cf.\thinspace \Cref{thm:main}) shows how the winning strategy $\ekvwin$ of Verifier in an $\epsilon$-acceptance game can be seen as the winning strategy of Verifier playing the traditional acceptance and bisimulation distance game simultaneously.  Furthermore, the two examples from \Cref{sec:optimal epsilon} suggest a strong connection between the optimal $\varepsilon$ value for the game and Lebesgue measures over trees.
In both examples, this optimal $\varepsilon$ value corresponds to the measure of the set of subtrees that are not valid from the point of view of the execution of the automaton.

\subparagraph{Future work}  Our work is built on the acceptance game \cite{kupke-venema:coalg-aut-theory} between an $F$-automaton and an $F$-coalgebra. As the first step, it is natural to extend the proof of main theorem to the level of coalgebras. This will allow one to recover various instances of $\epsilon$-acceptance games over variety of structures like arbitrary $n$-ary trees, probabilistic trees, etcetera. More concretely, our work suggests that distance lifting (by an appropriate choice of $\fun$) results in a measure on binary trees. In the future, this connection needs a detailed treatment; the exact relationship between our distance lifting and measures on binary tree is left as future work.

\printbibliography
\end{document}